%% file: main.tex
\documentclass[english,12pt]{article}

\usepackage[LGR,T1]{fontenc}
\usepackage[utf8]{inputenc}
\usepackage{pdflscape}
\usepackage{geometry}
\usepackage{array}
\usepackage{alltt}
\usepackage{float}
\usepackage{booktabs}
\usepackage{textcomp}
\usepackage{multirow}
\usepackage{amsmath}
\usepackage{amsthm}
\usepackage{amssymb}
\usepackage{authblk}
\usepackage{sectsty}
\usepackage{hyperref}
\usepackage{graphicx}
\usepackage{blindtext}
\usepackage{caption}
\usepackage{subcaption}
\usepackage{lmodern}
\usepackage{xcolor}
\usepackage{multicol}
\usepackage{tabularray}
\usepackage{autobreak}
\usepackage{algorithm}
\usepackage{algorithmic}
\usepackage{framed}
\usepackage{mathrsfs}
\usepackage{colortbl}
\usepackage{siunitx}  
\usepackage{threeparttable}
\usepackage{mathtools}

\definecolor{steelblue}{RGB}{70, 130, 180}

\usepackage{xcolor}

\hypersetup{
    colorlinks=true,
    linkcolor=blue,
    citecolor=blue,
    urlcolor=blue
}

\newcommand{\Vbb}{\mathbb{V}}
\newcommand{\indep}{\perp \! \! \! \perp}
\newcommand{\parrow}{\xrightarrow[]{p}}

\definecolor{drslide}{HTML}{A12210}
\definecolor{gslide}{HTML}{00B050}

\sectionfont{\fontsize{14}{13}\selectfont}
\subsectionfont{\fontsize{12}{12}\selectfont}

\newtheorem{theorem}{Theorem}
\newtheorem*{theorem*}{Theorem}  

\allowdisplaybreaks

\makeatletter

\ProvideTextCommand{\~}{LGR}[1]{\char126#1}

\theoremstyle{plain}

\usepackage{amsmath}

\makeatother

\usepackage{babel}
\providecommand{\theoremname}{Theorem}

\usepackage[backend=biber, style=authoryear, uniquename=init, giveninits, maxcitenames=2]{biblatex}

\DeclareDelimFormat{nameyeardelim}{\addcomma\space}

\DeclareCiteCommand{\parencite}[\mkbibparens]
  {\usebibmacro{prenote}}
  {\usebibmacro{citeindex}%
   \printtext[bibhyperref]{%
     \usebibmacro{cite}}}
  {\multicitedelim}
  {\usebibmacro{postnote}}

\DeclareCiteCommand{\cite}
  {\usebibmacro{prenote}}
  {\usebibmacro{citeindex}%
   \printtext[bibhyperref]{%
     \usebibmacro{cite}}}
  {\multicitedelim}
  {\usebibmacro{postnote}}

\let\cite\parencite

\usepackage{thmtools}
\declaretheoremstyle[
headfont=\normalfont\bfseries,
  notefont=\mdseries, notebraces={(}{)},
  bodyfont=\normalfont,
  spaceabove = 0.4cm, 
  spacebelow = 0.4cm,
  headpunct=\newline,    
  postheadspace=0pt 
]{mystyle}
\declaretheorem[style=mystyle]{example}

\usepackage{etoolbox}
\makeatletter
\patchcmd{\hyper@makecurrent}{%
    \ifx\Hy@param\Hy@chapterstring
        \let\Hy@param\Hy@chapapp
    \fi
}{%
    \iftoggle{inappendix}{
        \@checkappendixparam{chapter}%
        \@checkappendixparam{section}%
        \@checkappendixparam{subsection}%
        \@checkappendixparam{subsubsection}%
        \@checkappendixparam{paragraph}%
        \@checkappendixparam{subparagraph}%
    }{}%
}{}{\errmessage{failed to patch}}

\newcommand*{\@checkappendixparam}[1]{%
    \def\@checkappendixparamtmp{#1}%
    \ifx\Hy@param\@checkappendixparamtmp
        \let\Hy@param\Hy@appendixstring
    \fi
}
\makeatletter

\newtoggle{inappendix}
\togglefalse{inappendix}

\apptocmd{\appendix}{\toggletrue{inappendix}}{}{\errmessage{failed to patch}}

\begin{document}

\title{Improving Variance Estimation for Covariate Adjustment with Binary Outcomes}

\author[1,2$\dagger$]{Kaitlyn Lee}
\author[2]{Alex Ocampo}
\author[3]{Courtney Schiffman}
\author[4]{Michel Friesenhahn}
\author[3]{Christina Rabe}
\author[5]{Michael Rosenblum}

\affil[1]{University of California, Berkeley, Berkeley, CA, USA}
\affil[2]{Hoffmann La Roche Ltd, Basel, Switzerland}
\affil[3]{Genentech, South San Francisco, CA, USA}
\affil[4]{Michel Friesenhahn Consulting, LLC, Albany, CA, USA}
\affil[5]{Johns Hopkins University, Baltimore, MD, USA\vspace{1em}}

\affil[$\dag$]{Email: kaitlyn\_lee@berkeley.edu}
\setcounter{Maxaffil}{0}
\renewcommand\Affilfont{\itshape\small}
\maketitle
\date{}
\vspace{-11mm}
\begin{abstract}
Covariate adjustment is a general method for improving precision when estimating treatment effects in randomized trials and is recommended by the FDA in its 2023 guidance when baseline variables are prognostic for the primary outcome. We focus on a method highlighted in that guidance called ``standardization'' (or ``g-computation'') for estimating the marginal treatment effect. We address the question of how to reliably estimate variance for binary outcomes when marginal outcome probabilities are close to 0 or 1. We propose an influence function-based leave-one-out cross-validated (IF-LOO) variance estimator for the standardized difference-in-means average treatment effect. Through simulation studies, we show that this estimator provides appropriate type-I error control and performs reliably in challenging settings where existing methods can yield inflated type-I error or fail entirely, such as when outcome events are rare or sample sizes are small. In addition to having desirable statistical properties, we derive a closed-form expression for the proposed estimator, enabling straightforward and reliable implementation by study statisticians. The robust finite-sample performance and ease of implementation suggest the IF-LOO variance estimator is a prudent default choice for standardization in clinical trials. 
\end{abstract}

\section{Introduction}

Using baseline covariate information to increase efficiency and gain power in statistical analyses of randomized trials is a long-standing idea dating back to \textcite{fisher1934statistical}. One method for covariate adjustment developed in the causal inference literature is called ``standardization,'' and should be in the toolkit of every clinical trial statistician faced with analyzing a binary endpoint. This is not only because the method was highlighted by the FDA in their covariate adjustment guidance \cite{FDA2023covariates}, but also because it allows statisticians to perform a covariate adjusted analysis that targets the treatment effect as a marginal risk difference. Without such methods, targeting a marginal treatment effect would only be possible in an unadjusted analysis, leaving behind any potential efficiency gains. 

Standardization \cite{neison1844method,robins1986new} is also a more robust alternative to the traditional approach of adjusting for covariates in a logistic regression and performing inference on the parameter corresponding to the treatment covariate. Standardization begins by fitting the same traditional logistic model, but only uses this regression fit as a working model. From this model, one then predicts outcomes for each subject under either active or control treatment and compares the averages of these predictions to quantify the treatment effect. By averaging across all subjects, this ``standardizes'' the covariate distributions when comparing the estimated means of treated and untreated subjects. Thus, this adjustment procedure can be leveraged for efficiency gains - i.e., tighter confidence intervals and more informative p-values - when the modeled covariates are prognostic. Since the probability of the outcome is marginalized in each treatment arm separately before contrasting them, this allows the analyst to choose a risk difference or ratio as opposed to the less interpretable odds ratio. Another advantage is that the treatment effect estimate is asymptotically unbiased regardless of the postulated canonical link working model \cite{rosenblum2010simple}, so the treatment effect estimated by standardization is robust to model mispecification - i.e., it is semiparametric.

While the standardization procedure is well-established and provides a straightforward way to obtain a point estimate for the treatment effect, questions remain regarding the best way to estimate the variance of the estimator (and thus construct confidence intervals). The original draft FDA guidance released in 2021 referenced a variance estimation approach that proposed using the covariance estimate from the working logistic model and applies the delta method for the risk difference transformation \cite{ge2011covariate}. There exist several other methods for variance estimation, including influence-function plug-in estimators \cite{rosenblum2010simple}, methods based on the limiting distribution of the standardized estimator \cite{ye2023robust}, the nonparametric bootstrap \cite{efron_bootstrap_1979}, and M-estimation \cite{wang2023model}. 

Since the draft guidance on covariate adjustment was published, statisticians have begun comparing the performance of some of the above estimators. Most notably, a debate was sparked in part by the theoretical work of \textcite{ye2023robust}, which demonstrated that the variance estimation approach previously proposed by \textcite{ge2011covariate} neglects a component of the variance, leading to underestimation. Ye et al. present a simulation study comparing the new method proposed in the paper with the method proposed in Ge et al. and found that Ge et al.'s method had 91\% coverage while Ye et al.'s had correct 95\% coverage; however, it is unclear how applicable this result is to real-world usage, as the data-generating process for this simulation used a treatment effect corresponding to a very high odds ratio that may be unrealistic in clinical trial settings. In addition, this particular paper did not consider other variance estimators proposed in the literature. A more realistic simulation study was conducted by Liu and Xi which compared nine different variance estimators with much more comparable results among approaches, but still corroborating Ye's claims that those based the delta method have a tendency to provide anti-conservative inference \cite{liu2024covariate}. Ultimately, works such as these two papers resulted in the method in Ge et al. being removed from the final guidance \cite{FDA2023covariates}. 

However, notably absent from the literature are studies that have explicitly considered how the various variance estimators perform in scenarios with low event rates or small sample sizes, which are situations encountered in many randomized clinical trials. While exact procedures such as Fisher's \parencite*{fisher1934statistical} or the Suissa-Shuster \parencite*{suissa1985exact} test are well-formulated to handle the small sample setting for binary outcomes, they cannot reap the benefits of covariate adjustment. In our own explorations of the variance estimators for covariate adjustment when designing a phase III clinical trial, we uncovered that when event rates are low or the sample size is small, the proposed variance estimators for standardization either fail to obtain adequate coverage or suffer from computational drawbacks, which inspired us to develop a new approach. Therefore, in this manuscript, we propose a novel influence function-based leave-one-out (IF-LOO) estimate of the variance for the standardized treatment effect, which is an extension of a previously proposed influence function plug-in based variance estimator \cite{rosenblum2010simple}. This estimator draws on the history of using cross-fit and cross-validated estimators in semiparametric statistics (e.g., \cite{vanderLaanPolleyHubbard, Zheng2011, newey2018cross, chernozhukov_doubledebiased_2018, Balzer_vanderLaan_Petersen_2026}). It is also inspired by the leave-one-out strategy employed in \cite{colantuoni2015leveraging} which was originally used to assess whether adjusting for covariates may lead to substantial precision gains. 

The outline of this paper is as follows: \autoref{methods} reviews the methodology of standardization and describes our proposed variance approach while providing a proof of consistency. \autoref{results:sims} presents a simulation study in order to demonstrate adequate properties when binary event rates are low and the inappropriate coverage of other proposed influence-function based estimators of the variance. \autoref{discussion} concludes the paper with a discussion.

\section{Methods}
\label{methods}
We consider the randomized control trial setting as described in \cite{rosenblum2010simple}, where the outcome of interest is binary. Suppose our trial is comprised of $n$ participants. Let $\mathbf{X}$ denote a vector of baseline variables, $A$ denote the treatment assignment (where $A=1$ if the individual is assigned to the treatment arm and $A=0$ if the individual is assigned to the control arm), and $Y \in \{0, 1\}$ denote the outcome of interest.

We adopt the potential outcomes framework \cite{neyman1923application,rubin1974estimating}. Each individual has potential outcomes $Y_i(1)$ and $Y_i(0)$, corresponding to the outcome under treatment and control, respectively. We assume consistency, meaning that for an individual $i$, their observed outcome $Y_i = Y_i(A_i)$, the potential outcome corresponding to their treatment assignment. For each individual $i$, we observe the tuple $O_i = (\mathbf{X}_i, A_i, Y_i).$ We assume that the tuples $O_i$, $i = 1, \dots, n$ are independent, identically distributed draws from an unknown data generating process $\mathbb{P}_0$. We assume that all participants are randomly assigned to the treatment or control arm with probability $\mathbb{P}_0(A=1) = \pi_0$.

Our estimand of interest is the average treatment effect $\theta_{ATE}$, which, under the above assumptions, is equivalent to the difference in means in the two arms: \begin{align*}
    \theta_{ATE} &= \mathbb{E}_{\mathbb{P}_0}[Y(1)-Y(0)]\\
    &= \mathbb{E}_{\mathbb{P}_0}[Y(1)|A=1]-\mathbb{E}_{\mathbb{P}_0}[Y(0)|A=0]\\
    &= \mathbb{E}_{\mathbb{P}_0}[Y|A=1]-\mathbb{E}_{\mathbb{P}_0}[Y|A=0].
\end{align*}

\noindent Specifically, this equality holds because the expectation can be spread across the two counterfactuals due to the linearity of expectation, $A$ can be conditioned on because it is independent of $Y(a)$, $ \forall a\in \{0,1\}$ due to randomization, and the last equality holds due to consistency.

\subsection{Standardized Estimator}
\label{standardized_est}
For estimation, we employ the standardized estimator with a logistic regression working model described in \cite{rosenblum2010simple}. Let $m_{\boldsymbol{\beta}}(A, \mathbf{X})$ be a linear model including an intercept, coefficient on $A$, and possibly other terms involving covariates in $\mathbf{X}$ and $A$ (where $\boldsymbol{\beta}$ is the vector of coefficients). We posit the following working model:
\begin{align}
\label{eq:work_mod}
    \mathbb{P}(Y=1|A,\mathbf{X}) &= \frac{1}{1+\text{exp}(-m_{\boldsymbol{\beta}}(A,\mathbf{X}))}\notag \\
&= \text{expit}(m_{\boldsymbol{\beta}}(A,\mathbf{X})).
\end{align}

\begin{leftbar}
    \begin{example}[Main terms model]
    \label{ex:main_terms}
    For concreteness, we will use a model involving $A$ and $\mathbf{X}$ as main terms as a running example. More specifically, let $m_{\boldsymbol{\beta}}(A,\mathbf{X}) =  \beta_0 + \beta_A A + \boldsymbol{\beta}_{\mathbf{X}}^\top \mathbf{X}$. Then, our working model would be given by
    $$\mathbb{P}(Y=1|A,\mathbf{X}) = \text{expit}(\beta_0 + \beta_A A + \boldsymbol{\beta}_{\mathbf{X}}^\top \mathbf{X}).$$
\end{example}
\end{leftbar}

To estimate the ATE, researchers perform the following steps: 
\begin{enumerate}
\item Fit a logistic regression model $\widehat{\mu}(a,x) = \text{expit}(m_{\widehat{\boldsymbol{\beta}}}(a,x))$ on the data $O_1, \dots, O_n$, estimating the coefficients $\widehat{\boldsymbol{\beta}}$ using maximum likelihood estimation.
\item For each individual, predict potential outcomes $\widehat{Y_i}(1) = \widehat{\mu}(1,\mathbf{X}_i)$, $\widehat{Y_i}(0) = \widehat{\mu}(0,\mathbf{X}_i)$.
\item The standardized estimator is given by 
\begin{align}
\widehat{\theta}_{\text{standardized}} &= \frac{1}{n} \sum_{i=1}^n \widehat{Y_i}(1) - \widehat{Y_i}(0)\notag \\
& =\frac{1}{n} \sum_{i=1}^n \widehat{\mu}(1,\mathbf{X}_i)- \widehat{\mu}(0,\mathbf{X}_i).
\end{align}
\end{enumerate}

\begin{leftbar}
\renewcommand{\thmcontinues}[1]{Main terms model, cont.}
\begin{example}[continues=ex:main_terms]
Let $m_{\boldsymbol{\hat{\beta}}}(A, \mathbf{X}) = \widehat{\beta}_0 + \widehat{\beta}_A A + \widehat{\beta}_{\mathbf{X}}^\top \mathbf{X}$ be the linear part of the logistic regression, where the coefficients are obtained by maximum likelihood estimation in a logistic regression of the binary $Y_i$ on $A_i, \mathbf{X}_i$, and an intercept using observations $i = 1, \dots, n$. We write $\hat{\mu}(a, x) =$ expit$(m_{\boldsymbol{\hat{\beta}}}(a,x))$. Then, our estimator is given by
\begin{align*}
    \widehat{\theta}_{\text{standardized}}
    & =\frac{1}{n} \sum_{i=1}^n \widehat{\mu}(1,\mathbf{X}_i)- \widehat{\mu}(0,\mathbf{X}_i)\\
& =\frac{1}{n} \sum_{i=1}^n \text{expit}(m_{\widehat{\boldsymbol{\beta}}}(1,\mathbf{X}_i))- \text{expit}(m_{\widehat{\boldsymbol{\beta}}}(0,\mathbf{X}_i))\notag \\
&=\frac{1}{n} \sum_{i=1}^n \text{expit}\left(\widehat{\beta}_0 + \widehat{\beta}_A + \widehat{\beta}_{\mathbf{X}}^\top \mathbf{X}_i\right) - \text{expit}\left(\widehat{\beta}_0 + \widehat{\beta}_{\mathbf{X}}^\top \mathbf{X}_i\right).
\end{align*}
\end{example}
\end{leftbar}

As explained in \cite{rosenblum2010simple}, the working model in \autoref{eq:work_mod} \textit{need not be correctly specified} in the setting of a randomized controlled trial. Even when the working model is completely misspecified, the standardized estimator is still guaranteed to be consistent for $\theta_{ATE}$, meaning that the estimator converges to the true value as the sample size $n \to \infty$. 

\subsection{Variance Estimator for the Standardized Estimator}
\label{var}
As in \cite{rosenblum2010simple}, we assume there exists a maximizer $\boldsymbol{\beta}^*$ of the expected log-likelihood of the logistic regression, where the expectation is with respect to $\mathbb{P}_0$. Under this assumption, Rosenblum and van der Laan show that $\widehat{\boldsymbol{\beta}}$ converges to $\boldsymbol{\beta}^*$. We write $\mu^*(a,x) = \text{expit}(m_{\boldsymbol{\beta}^*}(a,x))$ to represent the logistic regression model with the limiting coefficients $\boldsymbol{\beta}^*$. They also show that the influence function of the standardized estimator is given by
\begin{align}
\label{method:standardized_if}
    \varphi(O) = \left(\frac{A}{\pi_0} - \frac{1-A}{1-\pi_0} \right)(Y-\mu^*(A,\mathbf{X})) +\mu^*(1,\mathbf{X}) - \mu^*(0,\mathbf{X}) - \theta_{ATE}.
\end{align} 

The authors then go on to show that the true asymptotic variance of the standardized estimator is given by 
$$\mathbb{V}_0 = \mathbb{E}_0[\varphi(O)^2].$$

\noindent The authors propose one estimator for the variance of $\widehat{\theta}_{\text{standardized}}$, which we call the \textit{plug-in influence function (plug-in IF)} estimator:
\begin{align}
\label{eq:var}
\widehat{\sigma}^2 = \widehat{\mathbb{V}}_{\text{plug-in IF}}/n
\end{align}
\noindent where
\begin{align}
\label{eq:plug-in=if}
\widehat{\mathbb{V}}_{\text{plug-in IF}} = \frac{1}{n} \sum_{i=1}^n \left(\left( \frac{A_i}{\pi_0} - \frac{1-A_i}{1-\pi_0} \right) \left(Y_i - \widehat{\mu}(A_i,  \mathbf{X}_i) \right) + \widehat{\mu}(1, \mathbf{X}_i) - \widehat{\mu}(0, \mathbf{X}_i) - \widehat{\theta}_{\text{standardized}} \right)^2
\end{align}
\noindent is an estimate of the asymptotic variance of the estimator.

This variance estimator is based on estimating the variance of the \textit{influence function} (IF) of the estimator and is consistent for the true variance of the estimator. An influence function measures the influence that a particular observation has on an estimator \cite{serfling1980,boos2013,}. Any regular and asymptotically linear estimator $\hat{\theta}$ admits a unique influence function expansion \cite{tsiatis2006} if there exists a function $\varphi$ with expectation zero such that:
\begin{align}
\label{if_expansion}
    \sqrt{n}\big(\hat{\theta} - \theta\big) = \frac{1}{\sqrt{n}}\sum_{i=1}^{n}\varphi(O_i) + o_{p}(1). 
\end{align}

Due to the connection between the expansion in \autoref{if_expansion} and the central limit theorem, the influence function $\varphi$ can then be used as a means of estimating the variance of an estimator:
\begin{align*}
     \widehat{\sigma}^2(\hat{\theta}) = \frac{1}{n^2}\sum_{i=1}^{n}\varphi^{2}(O_i). 
\end{align*}

In the space of semiparametric models, the \textit{efficient influence function} (EIF) is the influence function with the smallest asymptotic variance that reaches the semiparametric efficiency bound. That is,
\begin{align*}
     \Vbb(\varphi_{eff})\leq \Vbb(\varphi) \quad \forall \; \varphi  \in \mathcal{IF} 
\end{align*}
where $\mathcal{IF}$ is the set of all valid influence functions within the Hilbert-space. The influence function of $\widehat{\theta}_{standardized}$ in \autoref{method:standardized_if} is the efficient influence function $\varphi_{eff}$ and obtains this bound if the working model is correctly specified.

Despite being grounded in rigorous asymptotic statistical theory, as we show in \autoref{results:sims}, this IF-based variance estimator (along with others proposed in the current literature) can underestimate variance in certain finite-sample scenarios which researchers may plausibly face when running a randomized control trial, leading to invalid type-I error control. The main situations we focus on are 1) the sample size is small or 2) the binary outcome is rare in one or both arms.

\subsection{Main Result: Leave One Out IF Variance Estimation}

\label{var_loo_cv}

To overcome this finite-sample error, we propose a new variance estimator that builds upon the influence function based estimator above. This variance estimator employs a leave one out cross-validation (IF-LOO) approach to estimate the variance of the influence function.

Intuitively, one can think of the underestimation resulting from using the estimator $\widehat{\sigma}^2$ as an overfitting problem. In particular, the term $\left( \frac{A_i}{\pi_0} - \frac{1-A_i}{1-\pi_0} \right) \left(Y_i - \widehat{\mu}(A_i,  \mathbf{X}_i) \right)$ appears in \autoref{eq:plug-in=if}. When the sample size is small, or there is little variation in the observed outcomes, it is easy to see that our model may overfit to our particular data, meaning our estimates of the residuals $\left(Y_i - \widehat{\mu}(A_i,  \mathbf{X}_i) \right)$ are too small relative to the true residual $Y_i - \mu^*(A_i, X_i)$. By employing leave-one-out, we prevent such overfitting by ensuring the prediction for individual $i$ is generated using a model that was trained on all data except for that from individual $i$.

Let $\widehat{\mu}_{-i}(a,x)$ be the predictions from a fitted working logistic regression model using all of the available data \textit{excluding} individual $i$ ($n-1$ data points). Then, we define

\begin{equation}
\label{eq:var_loo_cv}
 \widehat{\sigma}^2_{\text{IF-LOO}}
=\frac{1}{n^2}
\sum_{i=1}^n \widehat{\varphi}_{-i}(O_i)^2,   
\end{equation} 

\noindent where
$$
\widehat{\varphi}_{-i}(O_i) = \left( \frac{A_i}{\pi_0} - \frac{1-A_i}{1-\pi_0} \right) \big( Y_i - \widehat{\mu}_{-i}(A_i, \mathbf{X}_i) \big) +
\widehat{\mu}_{-i}(1,\mathbf{X}_i) -
\widehat{\mu}_{-i}(0,\mathbf{X}_i) -\widehat{\theta}_{\text{standardized}}.
$$

\noindent In addition, we present the following theorem, stating that the proposed estimator is consistent:

\begin{theorem}
\label{theorem}
Let $\widehat{\mu}_{-i}$ be the MLE fit on observations $O_j:j\neq i$, $\widehat{\theta}_{\text{standardized}}=\frac{1}{n} \sum_{i=1}^n \widehat{\mu}(1,\mathbf{X}_i)- \widehat{\mu}(0,\mathbf{X}_i)$ be the full-data estimate for the ATE, and $\widehat{\varphi}_{-i}(O_i) = \left( \frac{A_i}{\pi_0} - \frac{1-A_i}{1-\pi_0} \right) \big( Y_i - \widehat{\mu}_{-i}(A_i, \mathbf{X}_i) \big) +
\widehat{\mu}_{-i}(1,\mathbf{X}_i) -
\widehat{\mu}_{-i}(0,\mathbf{X}_i) -\widehat{\theta}_{\text{standardized}}
$. Let 
$$\widehat{\mathbb{V}}_{\text{IF-LOO}} =\frac{1}{n}
\sum_{i=1}^n \widehat{\varphi}_{-i}(O_i)^2.$$

\noindent Under Conditions~1--2 in \autoref{app:conditions},
$$\widehat{\mathbb{V}}_{\text{IF-LOO}} \xrightarrow[]{p} \mathbb{V}_{0}.$$
\end{theorem}

\noindent We present the proof of this theorem in \autoref{app:proof}. We evaluate the finite sample properties in \autoref{results:sims} via simulation studies and show that the estimator results in much better coverage than other influence-function based methods for variance estimation.

\begin{leftbar}
    \renewcommand{\thmcontinues}[1]{Main terms model, cont.}
\begin{example}[continues=ex:main_terms]
For an individual $i$, let $m_{\widehat{\boldsymbol{\beta}}_{-i}}(A, \mathbf{X}) = \widehat{\beta}_{0,-i} + \widehat{\beta}_{A, -i} A + \widehat{\boldsymbol{\beta}}_{\mathbf{X}, -i}^\top \mathbf{X}$ be the estimated linear part of the logistic regression model using all observations excluding $(\mathbf{X}_i, A_i, Y_i)$, where the coefficients are obtained by maximum likelihood estimation in a logistic regression of $Y$ on $A, \mathbf{X}$, and an intercept. For each individual, calculate
\begin{align*}
\begin{autobreak}
        \widehat{\varphi}_{-i}(O_i) =
        \left( \frac{A_i}{\pi_0} - \frac{1-A_i}{1-\pi_0} \right) \big( Y_i  - \text{expit}(m_{\widehat{\boldsymbol{\beta}}_{-i}}(A_i, \mathbf{X}_i)) \big) 
        + \text{expit}(m_{\widehat{\boldsymbol{\beta}}_{-i}}(1, \mathbf{X}_i)) - \text{expit}(m_{\widehat{\boldsymbol{\beta}}_{-i}}(0, \mathbf{X}_i)) -\widehat{\theta}_{\text{standardized}}.
\end{autobreak}

\end{align*}

\noindent Then, our estimator is given by
\begin{align*}
\widehat{\sigma}^2_{\text{IF-LOO}}
=\frac{1}{n^2}
\sum_{i=1}^n \widehat{\varphi}_{-i}(O_i)^2.
\end{align*}
\end{example}
\end{leftbar}

\section{Simulation Studies}
\label{results:sims}

\subsection{Simulation Setup}

To evaluate the finite-sample performance of our estimator, we perform simulation studies mimicking realistic scenarios one may expect in a randomized clinical trial. Using simulated data allows us to evaluate the performance of the different estimators against a known ground truth. We consider two hypothetical randomized control trials with two data generating processes (DGPs) to highlight the types of scenarios in which researchers would gain the most from the IF-LOO variance estimator. The first is a scenario with a moderate sample size ($N=250$) and a small marginal placebo rate of the outcome in the population ($2.5\%$). The second is a scenario with a smaller sample size ($N=50$) and a larger marginal placebo rate of the outcome in the population ($25\%$).

In each of the two scenarios, each individual observation consists of six independent covariates: four continuous and two binary. The treatment is assigned via a random coin flip. The outcome is a binary variable where the log odds is linear in all six covariates and the treatment. We consider different average treatment effects of $0\%, 2.5\%, 5\%,10\%,$ and $15\%$. Since the ground truth for the marginal risk difference is not a parameter from the data-generating process, we utilize Gaussian quadrature to compute the true causal effect \cite{ocampo2026revealing, gauss1814methodus}. Each DGP is sampled 10,000 times, and each method is applied to each replicate. Explicit details on the simulation can be found in \autoref{app:dgps} and R code is included in the supplementary materials to foster reproducibility.


We consider two of the ATE estimators in each simulation: an unadjusted analysis and the standardized estimator. The unadjusted analysis is implemented by the difference-in-means estimator, with the variance calculated via a plug-in for the variance of a Bernoulli random variable. The standardized estimator is implemented by fitting a logistic regression working model on the data. We use a correctly specified model when implementing the standardized estimator. The reason for this is two-fold. As described above in \autoref{var}, if the working model is correctly specified, then the resulting standardized estimator has the lowest possible variance (in a semi-parametric sense) -- thus, using a correctly specified model is optimal in the sense of achieving the most precision gains possible out of covariate information. At the same time, this is the setting in which we expect the working model to overfit. Overfitting occurs when a model has many parameters relative to the amount of information in the data. With a binary outcome, that information comes mainly from participants whose outcome is less certain (a participant whose risk is near 0 is almost certain to have no event, and so tells the model little about the coefficients). Strongly prognostic covariates push most participants’ risks toward 0 or 1, so when events are rare or the sample is small, the effective sample size is closer to the number of events than to $n$. For example, the main-terms model in our first scenario with $n=250$ estimates eight coefficients from a trial that, under the null, has only about six events in expectation. With so little information per coefficient, the fitted model partly reproduces chance features of the observed outcomes rather than the true risks, and the estimated coefficients are themselves highly variable. The plug-in variance estimator treats the fitted model as if it were the true one, and so it misses both effects. Thus, the scenarios in which covariate adjustment is most beneficial are also those in which we would expect traditional variance estimators to perform worst.

As benchmark methods, we consider three estimators of the variance for the standardized estimator popular in the literature: the influence-function based estimator described in \cite{rosenblum2010simple} (which we refer to as the ``IF plug-in'' estimator), the variance estimator of \cite{ye2023robust} for the standardized estimator, as implemented in the RobinCar2 package \cite{robincar} (which we refer to as “RobinCar2”), and the nonparametric bootstrap (over 1000 bootstrap samples) \cite{efron_bootstrap_1979}. To evaluate the performance of the estimators, we consider the bias of the variance estimator as well as the coverage of the normal-based $95\%$ confidence interval.

\subsection{Simulation Results}

The main results of the simulation studies are presented in \autoref{tab:sim_n250} and \autoref{tab:sim_n50}, which display the empirical coverage of the nominal 
$95\%$ confidence interval across the five ATEs for each scenario.

\begin{table}[ht!]
\centering
\begin{threeparttable}
\caption{Simulation results ($N = 250$, placebo rate = 2.5\%) }
\label{tab:sim_n250}
\small
    \begin{tabular}{ll
      S[table-format=1.3]
      S[table-format=1.3]
      S[table-format=1.3]
      S[table-format=1.3]
      S[table-format=1.3]}
    \toprule
    & & \multicolumn{5}{c}{Method} \\
    \cmidrule(lr){3-7}
    {ATE} & {Metric} & {IF-LOO} & {IF Plug-In} & {RobinCar2} & {Bootstrap} & {Unadjusted} \\
    \midrule
    
    \multirow{3}{*}{0\%}
      & SE       & 0.0207 & 0.0207 & 0.0207 & 0.0207 & 0.0201 \\
      & Est.\ SD & 0.0198 & 0.0171 & 0.0172 & 0.0230 & 0.0194 \\
      & Type I Error & {5.52} & {11.83} & {11.80} & {3.56} & {5.29} \\
    \midrule
    
    \multirow{3}{*}{2.5\%}
      & SE       & 0.0237 & 0.0237 & 0.0237 & 0.0237 & 0.0243 \\
      & Est.\ SD & 0.0229 & 0.0207 & 0.0208 & 0.0247 & 0.0236 \\
      & Coverage & {94.35} & {91.14} & {91.43} & {95.64} & {94.14} \\
    \midrule
    
    \multirow{3}{*}{5\%}
      & SE       & 0.0258 & 0.0258 & 0.0258 & 0.0258 & 0.0275 \\
      & Est.\ SD & 0.0256 & 0.0237 & 0.0238 & 0.0265 & 0.0271 \\
      & Coverage & {94.57} & {92.15} & {92.37} & {95.53} & {94.17} \\
    \midrule
    
    \multirow{3}{*}{10\%}
      & SE       & 0.0298 & 0.0298 & 0.0298 & 0.0298 & 0.0323 \\
      & Est.\ SD & 0.0300 & 0.0283 & 0.0284 & 0.0301 & 0.0324 \\
      & Coverage & {94.65} & {92.99} & {93.27} & {95.10} & {95.00} \\
    \midrule
    
    \multirow{3}{*}{15\%}
      & SE       & 0.0333 & 0.0333 & 0.0333 & 0.0333 & 0.0368 \\
      & Est.\ SD & 0.0333 & 0.0317 & 0.0319 & 0.0334 & 0.0365 \\
      & Coverage & {94.83} & {93.56} & {93.67} & {94.98} & {94.24} \\
    
    \bottomrule
    \end{tabular}
    \begin{tablenotes}
    \small
    \item \textit{Note:} SE is the empirical standard deviation of the ATE point estimates over replicates.
    Type I error is the proportion of replicates ($\times 100$) for which the null hypothesis of no treatment effect is rejected at $\alpha = 0.05$.
    Coverage is the proportion of replicates ($\times 100$) for which the nominal 95\% CI contained the true ATE.
    A small number of replicates (0.05\% for ATE = 2.5\% and 0.51\% for ATE = 0\%) were excluded for all estimators due to errors thrown by the RobinCar2 package for variance estimation; replicates were excluded jointly across estimators so all estimators are compared on the same replicates.
    \end{tablenotes}
\end{threeparttable}
\end{table}

\begin{table}[ht!]
\centering
\begin{threeparttable}
\caption{Simulation results ($N = 50$, placebo rate = 25\%) }
\label{tab:sim_n50}
\small
    \begin{tabular}{ll
      S[table-format=1.3]
      S[table-format=1.3]
      S[table-format=1.3]
      S[table-format=1.3]
      S[table-format=1.3]}
    \toprule
    & & \multicolumn{5}{c}{Method} \\
    \cmidrule(lr){3-7}
    {ATE} & {Metric} & {IF-LOO} & {IF Plug-In} & {RobinCar2} & {Bootstrap} & {Unadjusted} \\
    \midrule
    
    \multirow{3}{*}{0\%}
      & SE       & 0.116 & 0.116 & 0.116 & 0.116 & 0.124 \\
      & Est.\ SD & 0.111 & 0.084 & 0.087 & 0.129 & 0.120 \\
      & Type I Error & {6.89} & {18.35} & {17.18} & {4.88} & {6.45} \\
    \midrule
    
    \multirow{3}{*}{2.5\%}
      & SE       & 0.116 & 0.116 & 0.116 & 0.116 & 0.124 \\
      & Est.\ SD & 0.112 & 0.086 & 0.088 & 0.130 & 0.123 \\
      & Coverage & {93.51} & {83.17} & {84.16} & {95.36} & {93.82} \\
    \midrule
    
    \multirow{3}{*}{5\%}
      & SE       & 0.118 & 0.118 & 0.118 & 0.118 & 0.128 \\
      & Est.\ SD & 0.112 & 0.087 & 0.090 & 0.130 & 0.124 \\
      & Coverage & {93.37} & {83.69} & {84.78} & {95.13} & {93.83} \\
    \midrule
    
    \multirow{3}{*}{10\%}
      & SE       & 0.117 & 0.117 & 0.117 & 0.117 & 0.131 \\
      & Est.\ SD & 0.114 & 0.090 & 0.092 & 0.130 & 0.127 \\
      & Coverage & {93.88} & {85.85} & {86.83} & {95.42} & {93.67} \\
    \midrule
    
    \multirow{3}{*}{15\%}
      & SE       & 0.117 & 0.117 & 0.117 & 0.117 & 0.132 \\
      & Est.\ SD & 0.115 & 0.092 & 0.095 & 0.128 & 0.129 \\
      & Coverage & {93.69} & {86.88} & {87.78} & {95.15} & {93.77} \\
    
    \bottomrule
    \end{tabular}
    \begin{tablenotes}
    \small
    \item \textit{Note:} SE is the empirical standard deviation of the ATE point estimates over replicates.
    Type I error is the proportion of replicates ($\times 100$) for which the null hypothesis of no treatment effect is rejected at $\alpha = 0.05$.
    Coverage is the proportion of replicates ($\times 100$) for which the nominal 95\% CI contained the true ATE.
    A small number of replicates (between 0.1\% and 0.6\%) were excluded for all estimators from all DGPs due to errors thrown by the RobinCar2 package; replicates were excluded jointly across estimators so all estimators are compared on the same replicates.
    \end{tablenotes}
\end{threeparttable}
\end{table}

We first consider moderate sample size scenario of $N=250$ with a low placebo rate of $2.5\%$ (\autoref{fig:n250}). IF-LOO  achieves coverage near the nominal level across all DGPs, as does the unadjusted analysis. Bootstrap exhibits slight overcoverage at low true ATE values, while IF Plug-In and RobinCar2 continue to undercover at low true ATE values, though the degree  of undercoverage is substantially attenuated relative to the $N=50$ scenario  and both converge toward the nominal level as the true ATE increases. This  attenuation is expected: as the sample size grows, the overfitting bias in the plug-in variance estimator shrinks, and the advantage of the LOO correction  diminishes. Nevertheless, the persistence of undercoverage for IF Plug-In and  RobinCar2 even at $N=250$ suggests that the correction offered by IF-LOO remains practically relevant at moderate sample sizes.

We next consider the more challenging scenario: a small sample size of $N=50$  with a placebo rate of $25\%$ (\autoref{fig:n50}). In this setting, IF Plug-In  and RobinCar2 exhibit substantial undercoverage across all DGPs, with type-I error as high as $18\%$ at the null. Coverage of the confidence intervals resulting from these two estimators improves only modestly as the true  ATE increases, never approaching the nominal level within the range of DGPs  considered. This is consistent with the theoretical motivation for the IF-LOO estimator: when there are few informative observations per coefficient, the plug-in variance estimator is downward biased due to overfitting of the outcome model, and this bias is severe enough to meaningfully distort inference. In contrast, the bootstrap estimator achieves coverage near the nominal level across all DGPs. The IF-LOO and unadjusted analysis both exhibit modest undercoverage, with coverage remaining stable between 93\% and 94\% across the range of true ATE values.

Taken together, the simulation results suggest that IF-LOO offers a meaningful improvement over the IF plug-in and RobinCar2 variance estimators in finite samples, particularly in the $N=50$ scenario. Bootstrap performs comparably to IF-LOO in terms of coverage. However, as we discuss further in \autoref{discussion}, the nonparametric bootstrap is not deterministic (i.e., it has Monte Carlo error). Running the nonparametric bootstrap with 1,000 bootstrap samples in the same data set multiple times can result in different estimates for the standard error. Increasing the number of bootstrap samples can decrease this source of error - however, it comes at increased computational time. Furthermore, working logistic regression models can fail to fit for some bootstrap samples where complete separation (all 0s or 1s in a particular covariate stratum) occurs. This leads to having to throw out some bootstrap samples, leading to an invalid bootstrap distribution. Because the IF-LOO variance has a closed-form expression, it avoids the computational instability and sampling errors associated with the bootstrap.

\begin{figure}[p]
    \centering
    \includegraphics[width=0.75\textwidth]{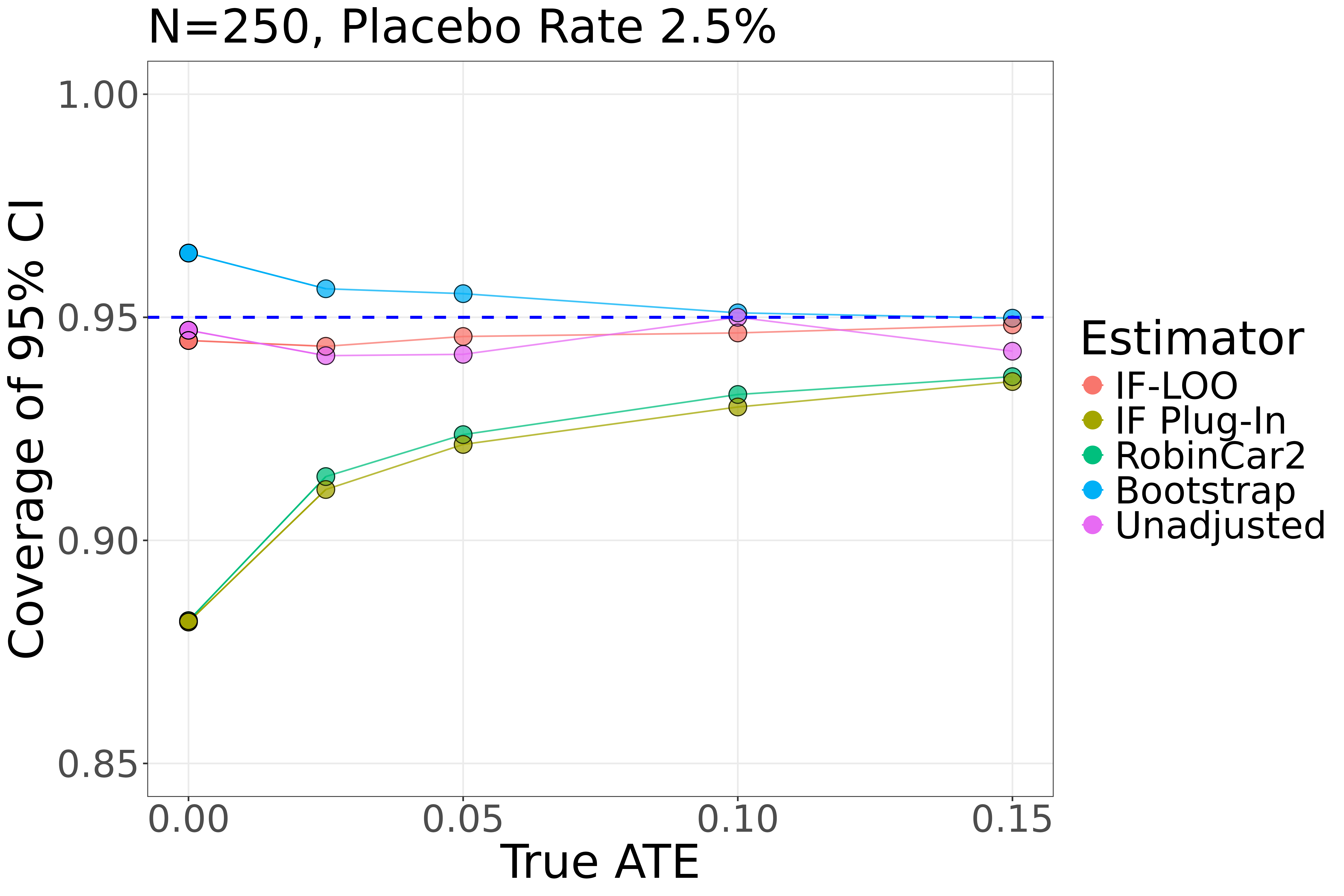}
    \caption{Simulation results for $N = 250$ with a placebo rate of $2.5\%$. The $x$-axis displays the true ATE and the $y$-axis displays the empirical coverage of the nominal $95\%$ confidence interval across 10{,}000 replicates, with the dashed blue line indicating the nominal $0.95$ level. A small number of replicates (0.05\% for ATE = 2.5\% and 0.51\% for ATE = 0\%) were excluded for all estimators due to numerical instability in matrix inversion during RobinCar2 variance estimation; replicates were excluded jointly across estimators so all estimators are compared on the same replicates. IF-LOO and the unadjusted analysis achieve coverage near the nominal level across all DGPs, while IF Plug-In and RobinCar2 exhibit undercoverage at low true ATE values, converging toward nominal coverage as the true ATE increases (but still undercovering). Bootstrap exhibits slight overcoverage at low true ATE values but returns to nominal levels at moderate and larger ATEs.}
    \label{fig:n250}
\end{figure}

\begin{figure}[p]
    \centering
    \includegraphics[width=0.75\textwidth]{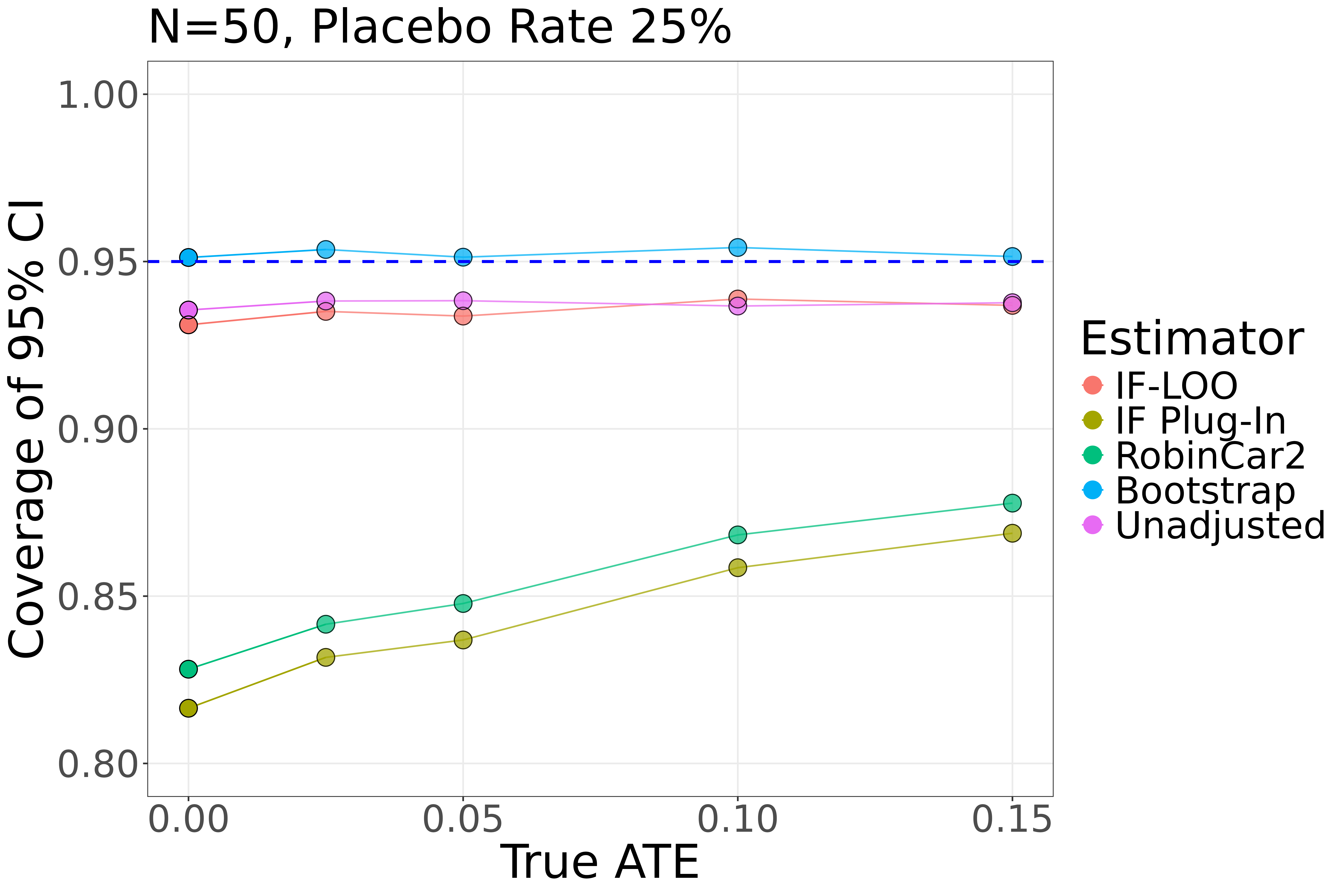}
    \caption{Simulation results for $N = 50$ with a placebo rate of $25\%$. The $x$-axis displays the true ATE and the $y$-axis displays the empirical coverage of the nominal $95\%$ confidence interval across 10{,}000 replicates, with the dashed blue line indicating the nominal $0.95$ level. A small number of replicates (between 0.1\% and 0.6\%) were excluded for all estimators due to numerical instability in matrix inversion during RobinCar2 variance estimation; replicates were excluded jointly across estimators so all estimators are compared on the same replicates. Bootstrap achieves coverage near the nominal level across all DGPs. IF-LOO and the unadjusted analysis exhibit modest undercoverage, remaining stable across true ATE values. IF Plug-In and RobinCar2 exhibit substantial undercoverage across all DGPs, with coverage increasing toward but not reaching the nominal level as the true ATE increases.}
    \label{fig:n50}
\end{figure}

\section{Discussion}
\label{discussion}

Traditionally, covariate adjustment for binary outcomes involved reporting the coefficient of treatment from a logistic regression model that incorporates treatment and covariates as an estimate for a treatment effect. However, this approach complicates the interpretation of the treatment effect due to the non-collapsibility of the odds ratio \cite{prentice1978regression,greenland1999confounding}. Non-collapsibility leads to the coefficients in the unadjusted and adjusted logistic regression models to actually target different estimands. This means that the statistician's choice of model specification will change the targeted treatment effect, which is not a desirable feature. Therefore, the simplest reason to prefer the marginal treatment effect to other conditional treatment effects is that, for a given dataset, there is only one marginal treatment effect, while there are many conditional effects, since there is an infinite choice of potential adjustment models. Marginal effects are also aligned with the treatment effect that is targeted when comparing outcomes between the arms of a study. The move towards targeting marginal effects on the risk difference scale in clinical trials is therefore welcomed and standardization provides a covariate-adjusted avenue for efficiently targeting these effects.

In this paper, we demonstrate the utility of a LOO influence function based variance estimator for the risk difference of a binary endpoint estimated via standardization. We demonstrate that the approach is well grounded in the asymptotic theory of influence functions and prove consistency of the proposed IF-LOO variance estimator. Via simulations, we demonstrate that in small samples and those with low event rates the proposed variance estimator provides more reliable coverage and close to nominal type-I error control as compared to standard methods proposed in the literature. Standard methods have a tendency to overfit the data in these situations (especially when there are many covariates relative to the number of observations). In small samples, even a few covariates can lead to overfitting and improper model specification. Based on the semi-parametric theory of the efficient influence function, the correctly specified model is the one that achieves the minimum proper variance (i.e., the semiparametric efficiency bound). The IF-LOO estimator seems to perform best especially in these scenarios. In large samples with non-rare event rates, the differences between the various variance estimators disappear and the issue is less of a concern.

The IF-LOO approach has a number of advantages. First, in clinical trial programming implementations where double-coding and precise matching of results is common practice, the LOO approach enables unambiguous exact matching because it does not involve any repeated sampling. Also, machine learning algorithms (e.g., ridge, lasso, etc.) could be easily employed for the plug-in predictions since the variance is influence function based, although as discussed above, more sophisticated functional forms run the risk of more overfitting, which is already a problem in these scenarios. Cross-validation or sample splitting is often necessary when employing machine learning algorithms, which while theoretically well-principled, has been shown to lead to convergence issues in finite samples \cite{shao2026benchmarking}.  Operationally, clinical trial analyses must often be double-coded, and a LOO based estimator has the luxury of being deterministic given the data - i.e., there is no random sampling involved, as with cross-validation. In general, this work helps us observe that with dichotomized data, unforeseen problems in covariate adjustment are more likely to arise. We also observed that with binary covariates, the undercoverage of overfit models tends to be more of an issue. Hence, we considered a mixture of binary and uniform distributed variables in our simulation study.

A number of avenues for future research are possible. Firstly, finite sample corrections \cite{tsiatis2008covariate} on the asymptotic-based variance estimates could help fix the undercoverage of these methods revealed in our simulation studies. Our approach is grounded in super-population inference, as this is the standard inferential framework for clinical trials; however, the relationship of these methods to randomization-based inference \cite{abadie2020sampling} should be explored more deeply. \textcite{magirr2025estimating} have noted that the inferential framework and desired estimand can inform the choice of variance estimator in the context of standardization. Relatedly, exact tests may also have the potential to integrate covariate adjustment via linearization of the test statistic, as done in \cite{sun2025improve}. We welcome motivated researchers to pursue these paths.

Overall, the IF-LOO variance estimator solves a number of issues currently faced in covariate adjustment with binary outcomes - especially when faced with low event rates or small sample sizes. Through both simulation and theory we demonstrate its advantageous properties and consider it a prudent default choice as a variance estimate to ensure valid inferences. 

\pagebreak

\printbibliography

\pagebreak

\include{appendix.tex}

\end{document}

%% file: appendix.tex
\appendix
\section{Proof of Theorem 1}
\label{app:proof}
In this appendix, we will prove \autoref{theorem}.  Throughout the proof, we use two different norms. The first is the $L_2$ norm. For a fixed function $f(a,\mathbf{x})$ (possibly estimated from the data), we write 

$$\|f(A,\mathbf{X})\|_2 = \mathbb{E}_0[f(A,\mathbf{X})^2]^{1/2},$$

\noindent the root mean square over a new draw of $(A, X)$ from $\mathbb{P}_0$ with $f$ held fixed.

The second is the empirical norm. For a vector $\vec{v} = [v_1, \dots, v_n]^\top$, we write

$$\|\vec{v}\|_n = \left(\frac{1}{n}\sum_{i=1}^n v_i^2\right)^{1/2}.$$

\noindent We first lay out a few necessary assumptions (in addition to those laid out in \autoref{methods}):

\subsection{Conditions}
\label{app:conditions}
\begin{enumerate}
   \item Simple randomization: $A \indep (\mathbf{X}, Y(0), Y(1))$ with known propensity score $\mathbb{P}_0(A=1)=\pi_0$. The propensity score is bounded: $\eta < \pi_0 < 1-\eta$ for $0 < \eta < 0.5$.
    \item Conditions on the working model: As in \autoref{standardized_est}, we require the working model to include an intercept, a main term for $A$, and possibly other terms involving $\mathbf{X}$ and $A$. The maximizer $\boldsymbol{\beta}^*$ exists and is unique, allowing us to define the oracle (possibly misspecified) outcome regression $\mu^*(a,x)$. We also assume that the model fit on the full $n$ data points via maximum likelihood estimation is consistent for the oracle outcome regression: $\|\hat{\mu}-\mu^*\|_{2} \parrow 0$.

\end{enumerate}    
The first condition is automatically met in a randomized control trial with simple randomization where the propensity score is bounded away from 0 and 1.
The second condition is the same as that used in \cite{rosenblum2010simple}. It says that the model we write has a certain form and must be consistent for the oracle $\mu^*(a,x)$ ($\boldsymbol{\beta}^*$ must be unique and $\mathbb{E}_0\big[\|\mathbf{X}\|\big] < \infty$) \cite{NEWEY19942111}.

\subsection{Preliminaries}

As in \autoref{var}, let $\mu^*(a,x) = \text{expit}(m_{\boldsymbol{\beta}^*}(a,x))$ be the logistic regression model with the limiting coefficients $\boldsymbol{\beta}^*$. As discussed in \autoref{var}, the IF of the standardized estimator is given by

\begin{align}
    \varphi(O) = \left(\frac{A}{\pi_0} - \frac{1-A}{1-\pi_0} \right)(Y-\mu^*(A,\mathbf{X})) +\mu^*(1,\mathbf{X}) - \mu^*(0,\mathbf{X}) - \theta_{ATE}.
\end{align}

Then the asymptotic variance of the standardized estimator is
$$
  \mathbb{V}_0 = \mathbb{E}_0\big[\varphi(O)^2\big],$$
in the sense that
$\sqrt{n}\,(\widehat\theta_{\text{standardized}} - \theta_{\text{ATE}})
\xrightarrow{d} \mathcal{N}(0, \mathbb{V}_0)$
\cite{rosenblum2010simple}.

\subsection{Proof}
First, we restate the theorem.
\begin{theorem*}
Under Conditions 1-2 in \autoref{app:conditions},
$$\widehat{\mathbb{V}}_{\text{IF-LOO}} \parrow \mathbb{V}_{0}.$$
\end{theorem*}

\noindent \textit{Proof.} We want to show
$$
\widehat{\mathbb{V}}_{\text{IF-LOO}} -  \mathbb{V}_{0}  \parrow 0.
$$

Let $\delta_i =  \varphi(O_i)- \widehat{\varphi}_{-i}(O_i) $. We can write

\begin{align*}

\begin{autobreak}
    \delta_i =
    \left(1-\frac{A_i}{\pi_0}\right)(\mu^*(1, \mathbf{X}_i)-\widehat{\mu}_{-i}(1, \mathbf{X}_i)) -\left(1-\frac{(1-A_i)}{1-\pi_0}\right)(\mu^*(0, \mathbf{X}_i)-\widehat{\mu}_{-i}(0, \mathbf{X}_i)) 
    +  \underbrace{\widehat{\theta}_{\text{standardized}} - \theta_{\text{ATE}}}_{\Delta_n}
\end{autobreak}\\
\begin{autobreak}
    =
    \left(1-\frac{A_i}{\pi_0}\right)(\mu^*(1, \mathbf{X}_i)-\widehat{\mu}_{-i}(1, \mathbf{X}_i)) -\left(1-\frac{(1-A_i)}{1-\pi_0}\right)(\mu^*(0, \mathbf{X}_i)-\widehat{\mu}_{-i}(0, \mathbf{X}_i)) 
    + \Delta_n.
\end{autobreak}
\end{align*}

We can rewrite $\widehat{\varphi}_{-i}(O_i)^2 = (\varphi(O_i)-\delta_i)^2$ and thus
\begin{align*}
    \widehat{\mathbb{V}}_{\text{IF-LOO}} -  \mathbb{V}_{0} &= \frac{1}{n} \sum_{i=1}^n (\varphi(O_i)-\delta_i)^2 - \mathbb{V}_{0}\\
    &\begin{autobreak}
        = \frac{1}{n} \sum_{i=1}^n \Big(\varphi(O_i)^2 - 2\varphi(O_i)\delta_i + \delta_i^2 \Big)- \mathbb{V}_{0} 
    \end{autobreak}\\
    &\begin{autobreak}
        = \underbrace{ \left(\sum_{i=1}^n \frac{\varphi(O_i)^2}{n} - \mathbb{V}_{0}\right)}_{\text{A}} - \underbrace{\frac{2}{n} \sum_{i=1}^n\varphi(O_i) \delta_i }_{\text{B}} + \underbrace{\frac{1}{n}\sum_{i=1}^n \delta_i^2}_{\text{C}}.
    \end{autobreak}
\end{align*}

\noindent We'll consider the three terms separately in more detail below in each subsection.

\subsubsection{Term A}
\label{app:termA}
\noindent Because $Y$ is binary and $\eta < \pi_0 < 1-\eta$, the influence function is bounded 

$$|\varphi| \leq \frac{1}{\eta}+2,$$
and thus $\mathbb{E}_0[\varphi^2] <\infty$. By the law of large numbers, we have 
$$\left(\sum_{i=1}^n \frac{\varphi(O_i)^2}{n} - \mathbb{V}_{0}\right)\parrow 0.$$

\subsubsection{Term C}
\label{app:termC}

We want to show 
$$\frac{1}{n} \sum_{i=1}^n \delta_i^2 \parrow 0.$$ Let $r_{-i}(a, x) = \mu^*(a, x)-\hat{\mu}_{-i}(a, x)$ be the estimation error of the model $\hat{\mu}_{-i}(a, x)$. 

We can write 

\begin{align*}
    \delta_i = \left(1-\frac{A_i}{\pi_0}\right)(r_{-i}(1, \mathbf{X}_i)) - \left(1-\frac{1-A_i}{1-\pi_0}\right)(r_{-i}(0, \mathbf{X}_i)) + \Delta_n.
\end{align*}

\noindent Thus, we can write the quantity of interest as a squared norm:

\begin{align*}
    \frac{1}{n} \sum_{i=1}^n \delta_i^2 &= \|\delta\|^2_n \\
\end{align*}

\noindent Write $s_i = (1-\frac{A_i}{\pi_0})r_{-i}(1,\mathbf{X}_i) - (1-\frac{1-A_i}{1-\pi_0})r_{-i}(0,\mathbf{X}_i)$ for the weighted estimation error, so that $\delta_i = s_i+\Delta_n$. By the triangle inequality, we can write

\begin{align*}
\|\delta\|_n &= \left\|s + \Delta_n\mathbf 1\right\|_n \\
&\leq \left\| s\right\|_n + \|\Delta_n\mathbf 1\|_n\\
&= \left\| s\right\|_n+ |\Delta_n|,
\end{align*}

\noindent where the final equality comes from the fact that $\|\Delta_n\mathbf 1\|_n = \left(\frac{1}{n} \sum_{i=1}^n \Delta_n^2\right)^{1/2} = |\Delta_n|$.

By \cite{rosenblum2010simple}, we know that $\Delta_n = O_p(n^{-1/2}),$ and thus $\Delta_n \parrow 0.$ So, it suffices to show that $$\left\| s\right\|_n\parrow 0,$$ or, equivalently (by the continuous mapping theorem), 
\begin{equation}
\label{eq:c_1_term}
    \left\| s\right\|_n^2\parrow 0.
\end{equation}

Note that $A$ is binary, and $\eta < \pi_0 <1-\eta$. Thus, both $\left|1-\frac{A}{\pi_0}\right|$ and $\left|1-\frac{1-A}{1-\pi_0}\right|$ are bounded by $\frac{1}{\eta}$. Using $(a+b)^2 \leq 2a^2 + 2b^2$ for real numbers $a,b$,

\begin{align*}
    \begin{autobreak}
    s_i^2 \leq 2
    \left(\left(1-\frac{A_i}{\pi_0}\right)(r_{-i}(1, \mathbf{X}_i))\right)^2 
    + 2\left(\left(1-\frac{1-A_i}{1-\pi_0}\right)(r_{-i}(0, \mathbf{X}_i))\right)^2
    \end{autobreak}\\
    \begin{autobreak}
    \leq 
    \frac{2}{\eta^2} \left(r_{-i}(1, \mathbf{X}_i)^2 +  r_{-i}(0, \mathbf{X}_i)^2\right).
    \end{autobreak}
\end{align*}

\noindent We can now bound the term of interest

\begin{align*}
    \left\| s\right\|_n^2 &\leq \frac{1}{n}\frac{2}{\eta^2}\sum_{i=1}^n \left[r_{-i}(1,\mathbf{X}_i)^2 + r_{-i}(0, \mathbf{X}_i)^2\right]
\end{align*}

Therefore, it suffices to show that 

\begin{align}
\label{eq:suff-cond}
    \frac{1}{n}\sum_{i=1}^n \left[r_{-i}(1,\mathbf{X}_i)^2 + r_{-i}(0, \mathbf{X}_i)^2 \right]\parrow 0
\end{align}

Consider the conditional expectation of $r_{-i}(1,\mathbf{X}_i)^2 + r_{-i}(0, \mathbf{X}_i)^2$ given the observations $O_{-i} = \{O_j: j \neq i\}$. Fixing $O_{-i}$, the error terms only depend on $O_{-i}$ through the regression $\hat{\mu}_{-i}$. Additionally, the observations are i.i.d., and thus $\mathbf{X}_i \indep O_{-i}$. Therefore, the expectation 
\begin{equation}
\label{eq:cond_exp}
    \mathbb{E}_0[r_{-i}(1,\mathbf{X}_i)^2 + r_{-i}(0, \mathbf{X}_i)^2|O_{-i}] = \mathbb{E}_{0,X}[r_{-i}(1, \mathbf{X})^2] + \mathbb{E}_{0,X}[r_{-i}(0, \mathbf{X})^2],
\end{equation}

\noindent where the expectation $\mathbb{E}_{0,X}[h(X)] \coloneqq \int h(x) d\mathbb{P}_{0,X}(x)$ for a fixed function $h$ (that may depend on the data). The right side is only random through the fit $\hat{\mu}_{-i}$. 

For a fixed $\hat{\mu}_{-i}$, we can write

\begin{align*}
    \left\| \hat{\mu}_{-i}-\mu^*  \right\|_{2}^2 &= \mathbb{E}_{0,X}[\mathbb{E}_0[r_{-i}(A, 
    \mathbf{X})^2|\mathbf{X}]]\\
    &=  \mathbb{E}_{0,X}[\mathbb{P}_0[A=1|\mathbf{X}] r_{-i}(1,\mathbf{X})^2] + \mathbb{E}_{0,X}[\mathbb{P}_0[A=0|\mathbf{X}]r_{-i}(0,\mathbf{X})^2]\\
    &\geq \eta\left(\mathbb{E}_{0,X}[r_{-i}(1,\mathbf{X})^2] + \mathbb{E}_{0,X}[r_{-i}(0,\mathbf{X})^2]\right),
\end{align*}

\noindent where the inequality in the third line is due to the bounds on the propensity score in Condition 1.

Substituting this back into \autoref{eq:cond_exp}, we can write

\begin{equation}
    \mathbb{E}_0[r_{-i}(1,\mathbf{X}_i)^2 + r_{-i}(0, \mathbf{X}_i)^2|O_{-i}] \leq \frac{1}{\eta} \|\widehat\mu_{-i} -\mu^*\|_{2}^2.
\end{equation}

Note that $\widehat\mu_{-i}$ is the MLE on $n-1$ i.i.d. observations. Therefore, $\|\widehat\mu_{-i}-\mu^*\|_2$ has the same distribution as $\|\widehat\mu-\mu^*\|_2$ if $\widehat\mu$ is fit with $n-1$ observations. Therefore, because $\|\hat{\mu}-\mu^*\|_{2} \parrow 0$ by Condition 2, we also have
\begin{align*}
    \|\hat{\mu}_{-i}-\mu^*\|_{2} \parrow 0.
\end{align*}
\noindent By the continuous mapping theorem, we also have

$$\|\hat{\mu}_{-i}-\mu^*\|_{2}^2 \parrow 0.$$

Now, note that $ \|\hat{\mu}_{-i}-\mu^*\|_{2}^2$ is bounded since $\hat{\mu}_{-i},\mu^* \in (0,1)$. Since $0 \leq \|\hat{\mu}_{-i}-\mu^*\|_{2}^2 \leq 1$ and $\|\hat{\mu}_{-i}-\mu^*\|_{2}^2 \parrow 0$, the dominated convergence theorem (which also holds under convergence in probability \cite{vanderVaart1998}) gives $$\mathbb{E}_0\left[ \|\hat{\mu}_{-i}-\mu^*\|_{2}^2\right]\rightarrow 0.$$

Consider again the sufficient condition in \autoref{eq:suff-cond}:
$$\frac{1}{n}\sum_{i=1}^n \left[r_{-i}(1,\mathbf{X}_i)^2 + r_{-i}(0, \mathbf{X}_i)^2 \right]\parrow 0.$$

\noindent We can relabel observations $i$ and 1 leaves the joint distribution of the i.i.d. observations $O_1, \dots, O_n$ unchanged. It also turns $\widehat\mu_{-i}$ into $\widehat\mu_{-1}$ because the MLE doesn’t depend on the order of the data. So $(\mathbf{X}_i,\widehat\mu_{-i}) \overset{d}{=} (\mathbf{X}_1,\widehat\mu_{-1})$ for all $i$. We can take the expectation:

\begin{align*}
    \mathbb{E}_0\left[\frac{1}{n}\sum_{i=1}^n \left[r_{-i}(1,\mathbf{X}_i)^2 + r_{-i}(0, \mathbf{X}_i)^2\right]\right] &= \mathbb{E}_0[r_{-1}(1,\mathbf{X}_1)^2 + r_{-1}(0, \mathbf{X}_1)^2]\\
    &= \mathbb{E}_0[\mathbb{E}_0[r_{-1}(1,\mathbf{X}_1)^2 + r_{-1}(0, \mathbf{X}_1)^2|O_{-1}]]\\
    &\leq \frac{1}{\eta}\mathbb{E}_0\left[ \|\hat{\mu}_{-1}-\mu^*\|_{2}^2\right]\\
    &\rightarrow 0.
\end{align*}

Since $\frac{1}{n} \sum_{i=1}^n \left[r_{-i}(1,\mathbf{X}_i)^2 + r_{-i}(0, \mathbf{X}_i)^2\right]$ is non-negative, by Markov's inequality, for any $\epsilon > 0$,

\begin{align*}
    \mathbb{P}_0\left(\left(\frac{1}{n} \sum_{i=1}^n \left[r_{-i}(1,\mathbf{X}_i)^2 + r_{-i}(0, \mathbf{X}_i)^2\right]\right) > \epsilon\right) &\leq \frac{1}{\epsilon} \mathbb{E}_0\left[\frac{1}{n} \sum_{i=1}^n \left[ r_{-i}(1,\mathbf{X}_i)^2 + r_{-i}(0, \mathbf{X}_i)^2\right]\right]\\
    &\rightarrow 0.
\end{align*}

\noindent Therefore, the average squared estimation error satisfies
$\frac{1}{n} \sum_{i=1}^n \left[r_{-i}(1, \mathbf{X}_i)^2 + r_{-i}(0, \mathbf{X}_i)^2 \right]\xrightarrow{p} 0$. By the bound above
\autoref{eq:suff-cond}, the weighted estimation error satisfies
$$
  \left\| s\right\|_n \le \frac{\sqrt{2}}{\eta}
  \left( \frac{1}{n} \sum_{i=1}^n r_{-i}(1, \mathbf{X}_i)^2 + r_{-i}(0, \mathbf{X}_i)^2 \right)^{1/2},
$$
and the right side converges to $0$ in probability by the continuous mapping
theorem, so $\left\| s\right\|_n \xrightarrow{p} 0$. By the triangle inequality,
$\|\delta\|_n \leq \|s\|_n + |\Delta_n|$, and both terms on the right converge to
$0$ in probability, so $\|\delta\|_n \xrightarrow{p} 0$. Squaring, again by the
continuous mapping theorem,
$$
  \frac1n \sum_{i=1}^n \delta_i^2 = \|\delta\|_n^2 \xrightarrow{p} 0,
$$
as desired.

\subsubsection{Term B}

\noindent By the Cauchy-Schwarz inequality, we can write

\begin{align*}
    \left|\frac{2}{n} \sum_{i=1}^n \delta_i \varphi(O_i) \right| &\leq  2 \left(\frac{1}{n}\sum_{i=1}^n \delta_i^2 \right)^{1/2} \left(\frac{1}{n}\sum_{i=1}^n \varphi(O_i)^2 \right)^{1/2}. 
\end{align*}

We already showed in \autoref{app:termC} that $\frac{1}{n}\sum_{i=1}^n \delta_i^2 \parrow 0$. Therefore, by the continuous mapping theorem, we have that the first factor converges in probability to 0: 
$$\left(\frac{1}{n}\sum_{i=1}^n \delta_i^2\right)^{1/2} \parrow 0 .$$

\noindent For the second factor, by the weak law of large numbers (as in \autoref{app:termA}),
$$\|\varphi\|_n^2 = \frac{1}{n} \sum_{i=1}^n \varphi(O_i)^2 \xrightarrow{p} \mathbb{V}_0,$$
and by the continuous mapping theorem
$$\|\varphi\|_n \xrightarrow{p} \mathbb{V}_0^{1/2}.$$

\noindent In particular, the
influence-function norm $\|\varphi\|_n$ is bounded in probability:
$$\|\varphi\|_n = O_p(1).$$ The product of a sequence converging to $0$ in
probability and a sequence bounded in probability converges to $0$ in
probability, so 
$$\left|\frac{2}{n} \sum_{i=1}^n \delta_i \varphi(O_i) \right| \le 2\, \|\delta\|_n \|\varphi\|_n \xrightarrow{p} 0.$$

\subsubsection{Combining all the terms}
Combining terms,
\begin{align*}
    \widehat{\mathbb{V}}_{\text{IF-LOO}} -  \mathbb{V}_{0} &\begin{autobreak}
        = \underbrace{\left(\sum_{i=1}^n \frac{\varphi(O_i)^2}{n} - \mathbb{V}_{0}\right)}_{
    \parrow 0} - \underbrace{\frac{2}{n} \sum_{i=1}^n \delta_i \varphi(O_i)}_{\parrow 0} + \underbrace{\frac{1}{n}\sum_{i=1}^n \delta_i^2}_{\parrow 0}
    \end{autobreak}\\
    &\parrow 0.
\end{align*}
\noindent Thus, $\widehat{\mathbb{V}}_{\text{IF-LOO}}$ is consistent.

\section{Data Generating Processes for Simulation Studies}
\label{app:dgps}
This section describes the data-generating process used in \autoref{results:sims}. All of the simulations presented in this paper have data generating processes of the same structure:

\begin{align*}
    X_1, X_2, X_3, X_4 &\overset{\mathrm{iid}}{\sim} Unif(0,1)\\
    X_5, X_6 &\overset{\mathrm{iid}}{\sim} Bernoulli(p=1/2)\\
    A &\sim Bernoulli(p=1/2)\\
    Y|A, \mathbf{X} &\sim Bernoulli(p = \text{expit}(\beta_0 + \beta_A A + \beta_{\mathbf{X}}^\top \mathbf{X}))\\
    \mathbf{X} &= [X_1, X_2, X_3, X_4, X_5, X_6]^\top\\
    \beta_{\mathbf{X}}  &= [2.5, 1.8, -2.8, -2.1, 2.0,  -2.0]
\end{align*}

\noindent \autoref{tab:dgp_parameters} describes the other parameters that change throughout the two scenarios presented:

\begin{table}[H]
\centering
\caption{Data Generating Process Parameters}
\label{tab:dgp_parameters}
\begin{tabular}{lllll}
\toprule
Placebo Rate & $N$ & $\beta_0$ & ATE & $\beta_A$ \\
\midrule
\multirow{5}{*}{2.5\%} & \multirow{5}{*}{250} & \multirow{5}{*}{$-4.9171$}
  & 0\%   & $0$ \\
  & & & 2.5\% & $0.8566$ \\
  & & & 5\%   & $1.4065$ \\
  & & & 10\%  & $2.1752$ \\
  & & & 15\%  & $2.7465$ \\
\midrule
\multirow{5}{*}{25\%} & \multirow{5}{*}{50} & \multirow{5}{*}{$-1.4828$}
  & 0\%   & $0$ \\
  & & & 2.5\% & $0.2028$ \\
  & & & 5\%   & $0.3967$ \\
  & & & 10\%  & $0.7643$ \\
  & & & 15\%  & $1.1131$ \\
\bottomrule
\end{tabular}
\end{table}

%% file: bibliography.bib
@techreport{FDA2023covariates,
  title       = {Adjusting for Covariates in Randomized Clinical Trials for Drugs and Biological Products},
  author      = {{U.S. Food and Drug Administration}},
  year        = {2023},
  month       = may,
  institution = {U.S. Department of Health and Human Services},
  type        = {Guidance for Industry},
  url         = {https://www.fda.gov/media/148910/download},
}

@article{ge2011covariate,
  title={Covariate-adjusted difference in proportions from clinical trials using logistic regression and weighted risk differences},
  author={Ge, Miaomiao and Durham, L Kathryn and Meyer, R Daniel and Xie, Wangang and Thomas, Neal},
  journal={Drug information journal: DIJ/Drug Information Association},
  volume={45},
  number={4},
  pages={481--493},
  year={2011},
  publisher={Springer}
}

@article{ye2023robust,
  title={Robust variance estimation for covariate-adjusted unconditional treatment effect in randomized clinical trials with binary outcomes},
  author={Ye, Ting and Bannick, Marlena and Yi, Yanyao and Shao, Jun},
  journal={Statistical theory and related fields},
  volume={7},
  number={2},
  pages={159--163},
  year={2023},
  publisher={Taylor \& Francis}
}

@article{rosenblum2010simple,
  title={Simple, efficient estimators of treatment effects in randomized trials using generalized linear models to leverage baseline variables},
  author={Rosenblum, Michael and {van der Laan}, Mark J},
  journal={The international journal of biostatistics},
  volume={6},
  number={1},
  pages={13},
  year={2010}
}

@article{magirr2025estimating,
  title={Estimating the Variance of Covariate-Adjusted Estimators of Average Treatment Effects in Clinical Trials With Binary Endpoints},
  author={Magirr, Dominic and Wang, Craig and Przybylski, Alexander and Baillie, Mark},
  journal={Pharmaceutical Statistics},
  volume={24},
  number={4},
  pages={e70021},
  year={2025},
  publisher={Wiley Online Library}
}

@article{liu2024covariate,
  title={Covariate adjustment and estimation of difference in proportions in randomized clinical trials},
  author={Liu, Jialuo and Xi, Dong},
  journal={Pharmaceutical Statistics},
  volume={23},
  number={6},
  pages={884--905},
  year={2024},
  publisher={Wiley Online Library}
}

@article{robins1986new,
  title={A new approach to causal inference in mortality studies with a sustained exposure period—application to control of the healthy worker survivor effect},
  author={Robins, James M},
  journal={Mathematical Modelling},
  volume={7},
  number={9-12},
  pages={1393--1512},
  year={1986},
  publisher={Elsevier}
}

@book{fisher1934statistical,
  author    = {Fisher, Ronald Aylmer},
  title     = {Statistical Methods for Research Workers},
  edition   = {5th},
  year      = {1934},
  publisher = {Oliver and Boyd},
  address   = {Edinburgh and London},
  pages     = {xiii+319}
}

@article{prentice1978regression,
  title={Regression analysis of grouped survival data with application to breast cancer data},
  author={Prentice, Ross L and Gloeckler, Lynn A},
  journal={Biometrics},
  pages={57--67},
  year={1978},
  publisher={JSTOR}
}

@article{greenland1999confounding,
  title={Confounding and collapsibility in causal inference},
  author={Greenland, Sander and Pearl, Judea and Robins, James M},
  journal={Statistical science},
  volume={14},
  number={1},
  pages={29--46},
  year={1999},
  publisher={Institute of Mathematical Statistics}
}

@article{neison1844method,
  title={On a method recently proposed for conducting inquiries into the comparative sanitary condition of various districts},
  author={Neison, F. G. P.},
  journal={Journal of the Statistical Society of London},
  volume={7},
  number={1},
  pages={40--68},
  year={1844},
  publisher={JSTOR}
}

@article{abadie2020sampling,
  title={Sampling-based versus design-based uncertainty in regression analysis},
  author={Abadie, Alberto and Athey, Susan and Imbens, Guido W and Wooldridge, Jeffrey M},
  journal={Econometrica},
  volume={88},
  number={1},
  pages={265--296},
  year={2020},
  publisher={Wiley Online Library}
}

@article{ocampo2026revealing,
  title={Revealing the Truth: Calculating True Values in Causal Inference Simulation Studies via Gaussian Quadrature},
  author={Ocampo, Alex and Giudice, Enrico and McCaw, Zachary R and Morris, Tim P},
  journal={arXiv preprint arXiv:2601.05128},
  year={2026}
}

@book{gauss1814methodus,
  title={Methodus nova integralium valores per approximationem inveniendi},
  author={Gauss, Carl Friedrich},
  year={1814}
}

@article{tsiatis2008covariate,
  title={Covariate adjustment for two-sample treatment comparisons in randomized clinical trials: a principled yet flexible approach},
  author={Tsiatis, Anastasios A and Davidian, Marie and Zhang, Min and Lu, Xiaomin},
  journal={Statistics in medicine},
  volume={27},
  number={23},
  pages={4658--4677},
  year={2008},
  publisher={Wiley Online Library}
}

@article{wang2023model,
  title={Model-robust inference for clinical trials that improve precision by stratified randomization and covariate adjustment},
  author={Wang, Bingkai and Susukida, Ryoko and Mojtabai, Ramin and Amin-Esmaeili, Masoumeh and Rosenblum, Michael},
  journal={Journal of the American Statistical Association},
  volume={118},
  number={542},
  pages={1152--1163},
  year={2023},
  publisher={Taylor \& Francis}
}

@article{suissa1985exact,
  title={Exact unconditional sample sizes for the 2 times 2 binomial trial},
  author={Suissa, Samy and Shuster, Jonathan J},
  journal={Journal of the Royal Statistical Society: Series A (General)},
  volume={148},
  number={4},
  pages={317--327},
  year={1985},
  publisher={Wiley Online Library}
}

@article{sun2025improve,
  title={Improve the Precision of Area Under the Curve Estimation for Recurrent Events Through Covariate Adjustment},
  author={Sun, Jiren and Wang, Tuo and Yi, Yanyao and Ye, Ting and Shao, Jun and Du, Yu},
  journal={Statistics in Medicine},
  volume={44},
  number={15-17},
  pages={e70187},
  year={2025},
  publisher={Wiley Online Library}
}

@article{shao2026benchmarking,
  title={Benchmarking covariate-adjustment strategies for randomized clinical trials},
  author={Shao, Yulin and Lyu, Liangbo and Yu, Menggang and Wang, Bingkai},
  journal={arXiv preprint arXiv:2602.00434},
  year={2026}
}

@article{colantuoni2015leveraging,
  title={Leveraging prognostic baseline variables to gain precision in randomized trials},
  author={Colantuoni, Elizabeth and Rosenblum, Michael},
  journal={Statistics in medicine},
  volume={34},
  number={18},
  pages={2602--2617},
  year={2015},
  publisher={Wiley Online Library}
}

@article{neyman1923application,
  title={On the application of probability theory to agricultural experiments. Essay on principles. Section 9.},
  author={Neyman, Jerzy},
  journal={Statistical Science},
  pages={465--480},
  year={1923},
  publisher={JSTOR}
}

@book{serfling1980,
    author = {Serfling, RJ},
    title = {Approximation Theorems of Mathematical Statistics},
    year = {1980},
    publisher = {J. Wiley \& Sons},
    doi = {10.1002/9780470316481}
}

@book{boos2013,
    author = {Boos, DD and Stefanski, LA},
    title = {Essential Statistical Inference: Theory and Methods},
    year = {2013},
    publisher = {Springer},
    doi = {10.1007/978-1-4614-4818-1}
}

@book{tsiatis2006,
    author = {Tsiatis, AA},
    title = {Semiparametric Theory and Missing Data},
    publisher = {Springer},
    year = {2006},
    edition = {1},
    doi = {10.1007/0-387-37345-4}
}

@article{rubin1974estimating,
  title={Estimating causal effects of treatments in randomized and nonrandomized studies.},
  author={Rubin, Donald B},
  journal={Journal of educational Psychology},
  volume={66},
  number={5},
  pages={688},
  year={1974},
  publisher={American Psychological Association}
}

@Manual{robincar,
  title = {RobinCar2: ROBust INference for Covariate Adjustment in Randomized Clinical Trials},
  author = {Liming Li and Marlena Bannick and Daniel {Sabanes Bove} and Dong Xi and Ting Ye and Yanyao Yi},
  year = {2026},
  note = {R package version 0.2.2},
  url = {https://github.com/openpharma/RobinCar2/},
}

@book{vanderVaart1998,
  title={Asymptotic Statistics},
  author={{van der Vaart}, A. W. },
  year={1998},
  publisher={Cambridge University Press }
}

@article{efron_bootstrap_1979,
	title = {Bootstrap {Methods}: {Another} {Look} at the {Jackknife}},
	volume = {7},
	issn = {0090-5364, 2168-8966},
	shorttitle = {Bootstrap {Methods}},
	url = {https://projecteuclid.org/journals/annals-of-statistics/volume-7/issue-1/Bootstrap-Methods-Another-Look-at-the-Jackknife/10.1214/aos/1176344552.full},
	doi = {10.1214/aos/1176344552},
	number = {1},
	journal = {The Annals of Statistics},
	publisher = {Institute of Mathematical Statistics},
	author = {Efron, B.},
	month = jan,
	year = {1979},
	pages = {1--26},
}

@article{Balzer_vanderLaan_Petersen_2026, title={Machine learning to optimize precision in the analysis of randomized trials: A journey in pre-specified, yet data-adaptive learning}, DOI={10.1177/17407745261417227}, journal={Clinical Trials}, author={Balzer, Laura B and {van der Laan}, Mark J and Petersen, Maya L}, year={2026}, month={Feb}}

@article{vanderLaanPolleyHubbard,
url = {https://doi.org/10.2202/1544-6115.1309},
title = {Super Learner},
author = {Mark J. {van der Laan} and Eric C Polley and Alan E. Hubbard},
volume = {6},
number = {1},
journal = {Statistical Applications in Genetics and Molecular Biology},
doi = {doi:10.2202/1544-6115.1309},
year = {2007},
lastchecked = {2026-05-05}
}

@article{newey2018cross,
  title={Cross-fitting and fast remainder rates for semiparametric estimation},
  author={Newey, Whitney K and Robins, James R},
  journal={arXiv preprint arXiv:1801.09138},
  year={2018}
}

@Inbook{Zheng2011,
author="Zheng, Wenjing
and {van der Laan}, Mark J.",
title="Cross-Validated Targeted Minimum-Loss-Based Estimation",
bookTitle="Targeted Learning: Causal Inference for Observational and Experimental Data",
year="2011",
publisher="Springer New York",
address="New York, NY",
pages="459--474",
isbn="978-1-4419-9782-1",
doi="10.1007/978-1-4419-9782-1_27",
url="https://doi.org/10.1007/978-1-4419-9782-1_27"
}

@article{chernozhukov_doubledebiased_2018,
	title = {Double/debiased machine learning for treatment and structural parameters},
	volume = {21},
	issn = {1368-4221},
	url = {https://doi.org/10.1111/ectj.12097},
	doi = {10.1111/ectj.12097},
	number = {1},
	journal = {The Econometrics Journal},
	author = {Chernozhukov, Victor and Chetverikov, Denis and Demirer, Mert and Duflo, Esther and Hansen, Christian and Newey, Whitney and Robins, James},
	month = feb,
	year = {2018},
	note = {\_eprint: https://academic.oup.com/ectj/article-pdf/21/1/C1/27684918/ectj00c1.pdf},
	pages = {C1--C68},
}

@incollection{NEWEY19942111,
title = {Chapter 36 Large sample estimation and hypothesis testing},
series = {Handbook of Econometrics},
publisher = {Elsevier},
volume = {4},
pages = {2111-2245},
year = {1994},
issn = {1573-4412},
doi = {https://doi.org/10.1016/S1573-4412(05)80005-4},
url = {https://www.sciencedirect.com/science/article/pii/S1573441205800054},
author = {Whitney K. Newey and Daniel McFadden}
}
